\documentclass[a4paper,11pt]{article}

\usepackage[headinclude]{typearea}
\usepackage[english]{babel}           
\usepackage[utf8]{inputenc}
\usepackage{amsmath,amsthm,amsfonts,amssymb}
\usepackage {color}                                  
\usepackage {pifont}                                 
\usepackage {multicol}                               
\usepackage[numbers]{natbib}                                                     
\usepackage {graphicx}                               
\usepackage {array}                                  
\usepackage {longtable}                              
\usepackage {makeidx}                                
\makeindex
\usepackage{fancyvrb}
\usepackage{enumerate} 
\usepackage{ifpdf}
\usepackage{caption}
\usepackage{mathrsfs}
\usepackage[showonlyrefs]{mathtools}

\usepackage{setspace}
\usepackage{url}

\usepackage{marvosym}

\theoremstyle{plain}
\newtheorem{theorem}{Theorem}[section]

\newtheorem{lemma}[theorem]{Lemma}

\theoremstyle{definition}
\newtheorem{assumption}{Assumption}[section]
\newtheorem{definition}[theorem]{Definition}
\newtheorem{example}[theorem]{Example}
\newtheorem{remark}[theorem]{Remark}

\newcommand{\exclude}[1]{}

\newcommand{\E}{{\mathbb{E}}}

\newcommand{\N}{{\mathbb{N}}}
\renewcommand{\P}{{\mathbb{P}}}

\newcommand{\R}{{\mathbb{R}}}

\newcommand{\B}{{\mathbb{B}}}

\definecolor{darkgreen}{rgb}{0,0.5,0}

\definecolor{rot}{rgb}{0.75,0,0.25}

\definecolor{magenta}{rgb}{0.75,0,0.25}

\definecolor{violet}{rgb}{0.25,0,0.75}

\definecolor{softgreen}{rgb}{0,0.7,0.3}

\newcommand{\re}{\operatorname{Re}}
\newcommand{\im}{\operatorname{Im}}

\renewcommand{\P}{{\mathbb P}}

\newcommand{\cA}{{\cal A}}

\newcommand{\cC}{{\cal C}}

\newcommand{\cF}{{\cal F}}

\newcommand{\cH}{{\cal H}}

\newcommand{\cK}{{\cal K}}

\newcommand{\cM}{{\cal M}}

\newcommand{\cP}{{\cal P}}
\newcommand{\cR}{{\cal R}}

\newcommand{\cU}{{\cal U}}

\newcommand{\intens}{\lambda}
\newcommand{\state}{\R_0^+ \times \R}
\newcommand{\stateT}{[0,T] \times \R_0^+ \times \R}

\title{Utility Indifference Pricing of Insurance Catastrophe Derivatives}

\author{Andreas Eichler, Gunther Leobacher, Michaela Sz\"olgyenyi}

\begin{document}

\date{Preprint, May 2017}

\maketitle


\begin{abstract}
We propose a model for an insurance loss index 
and the claims process of a single insurance company holding a fraction
of the total number of contracts
that captures both ordinary losses and losses due to catastrophes.
In this model we price a catastrophe derivative by the method of utility indifference pricing.
The associated stochastic optimization problem is treated by techniques for piecewise
deterministic Markov processes.
A numerical study illustrates our results.\\
 
\end{abstract}


\section{Introduction}
\label{sec:introduction}

Costly natural catastrophes in the recent past
(hurricane Andrew in 1992,
hurricane Katrina in 2005,
the earthquake and tsunami in Japan 2011 resulting in the nuclear disaster at Fukushima,
floods in Thailand 2011)
all caused severe stress to the (re-)insurance industry.
However, these losses are still
small relative to losses of the US stock and bond markets. 
Therefore securitization (i.e.~transferring part of the risk to the financial market) is an efficient alternative to reinsuring catastrophe (CAT) losses, cf.~\cite{clp}.

Contracts of this kind are insurance-linked derivatives.\footnote{Details on currently listed insurance-linked derivatives can be found at \url{www.artemis.bm/deal_directory}.}
They are usually written on insurance industry
catastrophe loss indices, insurer-specific catastrophe losses, or parametric
indices based on the physical characteristics of catastrophe events. 
We focus on the first kind of products; they involve more basis
risk, but are less exposed to moral hazard than the others, cf.~\cite{cum}.

Derivatives written on insurance industry catastrophe loss indices were first issued in 1992 by the Chicago Board of Trade; these were futures and later also call- and put spread options written on aggregate CAT-loss indices, cf.~\cite{cum}.

A call spread option is a combination of a call option long and a call option short with a higher strike.
Another  popular type of catastrophe derivative is the CAT bond.
This is a classical bond combined with an option that is triggered by a (predefined) catastrophe event.
Note that the buyer of the bond thereby sells the embedded option.
The issuer is typically a (re-)insurance company that wants to reinsure parts of its risk exposure on the financial market.
In return the investor receives a coupon.

CAT derivatives are interesting for investors who seek to diversify their risk, since they are largely uncorrelated with classical financial instruments.\\

The challenges in pricing CAT derivatives are that the underlying index is not a traded asset,
that they are not liquidly traded themselves and, maybe most of all, the modeling of catastrophe events.\\

In the following we review the existing literature. For a more detailed literature overview we refer to \citet{muermann2004}.

\citet{Geman} study European vanilla call options written on an insurance loss index, which is
modeled by a jump-diffusion.
\citet{Cox} model the aggregate loss of an insurance company by a Poisson
process with constant arrival rate of catastrophe events
and derive a pricing formula for CAT-puts.
\citet{Jaimungal} model the aggregate loss by a
compound Poisson process to describe the dynamic losses more accurately.
\citet{muermann2008} derives the market price of insurance risk from CAT derivative prices in a compound Poisson model.
\citet{LeoNgare} use the method of utility indifference pricing to price CAT derivatives
written on an insurance loss index modeled by a compound Poisson process.

For catastrophe events, the assumption that the resulting claims
occur at jump times of a Poisson process as adopted by most previous studies is
not beyond justifiable critique.
A generalization was proposed in \citet{embrechts1997}, who model an insurance loss index by a doubly stochastic Poisson process (Cox process), i.e. the arrival rate of claims is a stochastic process itself; they price CAT futures in this model. \citet{Lin} also model the arrival of CAT events by a doubly stochastic Poisson process.
See also \citet{fuji} for no-arbitrage pricing of  CAT bonds in this context.
\citet{dassios2003} study the valuation of CAT derivatives by risk neutral valuation, where the underlying is modeled as a Cox process with shot noise intensity.\\

In this paper we introduce a novel model for an insurance loss index and for a single insurance portfolio that captures ordinary insurance losses as well as catastrophe losses.
We model the ordinary claims in the loss index by a compound Poisson process with constant intensity and we model the arrival of catastrophes by a Poisson process with constant intensity, where a jump triggers another stochastic variable that determines the number of claims in case of a catastrophe.

The claims process of a single insurance company holding a fraction
of the total number of contracts is then a \emph{dynamic thinning} of the process describing the index.
Our model has the advantage that the jump height distribution does not need to capture both many small claims and outliers caused by catastrophes, but these outliers are split into many smaller claims.
The dynamic thinning is reached in a very convenient way (by drawing from a uniform distribution on $[0,1]$) that we believe to be applicable in many other situations.\\

Using this model we present a pricing mechanism for CAT derivatives (like CAT spread options).
Since the insurance loss index is not a tradable asset, and since the market for CAT derivatives is not liquid, risk neutral valuation is not applicable. Instead we use the method of utility indifference pricing.
For this we need a hedging mechanism, which will be an active management of the risk portfolio.
The pricing method requires solving an associated stochastic optimization problem.

Our paper extends \cite{LeoNgare} by a more realistic modeling approach for the insurance loss index
and also for the thinning. In our paper a catastrophe event may partly hit the considered insurance company, whereas in \cite{LeoNgare} a catastrophe event always only affects one company.
Their model for the claims process of a single insurance company is a thinning (a change of the intensity) of the Poisson process driving the number of claims, while ours is a dynamic thinning of the claims for each event and thus has a different distribution of the jumps.

The model presented here is technically harder to handle; we provide the mathematical toolkit in this paper.
Using this new model instead of a simpler one is justified by our numerical results, which show that the new model has a significant impact on the price of a CAT derivative as it reflects catastrophes more accurately.

We also introduce a way to compute the utility indifference price of the derivative by Fourier techniques.
This method also allows to compute the residual risk and the profit-loss distribution and therefore to evaluate coherent risk measures.\\

The paper is organized as follows.
In Section \ref{sec:modeling} we model the insurance loss index 
and the claims process of a single insurance company.
The state process based on which the CAT derivative is priced, is identified as a
piecewise deterministic Markov process (PDMP), see \cite{davis1993, BR2011}.
In Section \ref{sec:pricing} we recall
the general concept of utility indifference pricing and we  
solve the associated stochastic optimization problem.
In Section \ref{sec:numerics} we show how the utility indifference price and also quantities relevant for risk management can be computed efficiently, and we  present a numerical study.


\section{The model}
\label{sec:modeling}

Let $(\Omega,\cF,\P)$ be a probability space carrying all stochastic variables appearing below.

Suppose we have a global claims process $C$, which keeps
track of all property insurance claims in a given country and we consider
an insurance company in the same country, so that the index will contain  the
losses of that particular insurance company among others.

The portfolio income rate consisting of the premium revenues from the risk portfolio is given by a continuous function $q$ of the company's market share $\xi \in [0,1]$.
The function $q$ is not necessarily linear in $\xi$, since demand for insurance might depend on the premium the company charges.
The wealth process of the insurance company can 
be controlled by managing the insured portfolio, i.e.~by controlling the market share $\xi$.
This allows for optimizing the management strategy for maximizing utility from terminal wealth.
Therefore, we can apply the method of utility indifference pricing for the valuation of CAT-derivatives.\\

The global claims process is given by
\begin{align}\label{eq:claimsprocess0}
C_t=\sum_{i=1}^{N^1_t} Y_{i,1} +\sum_{i=1}^{N^2_t} Z_i\,, \,\,\text{ where }\,\,
Z_i=\sum_{j=2}^{\tilde A_i}Y_{i,j}\,,
\end{align}
and where $N^1=(N^1_t)_{t\ge0}$ and $N^2=(N^2_t)_{t\ge0}$ are independent Poisson processes with intensities
$\lambda^1,\lambda^2$ and jump times $(\tau^1_{i})_{i\ge 1},(\tau^2_{i})_{i\ge 1}$.
The jump heights $(Y_{i,j})_{i,j\ge 1}$ are iid random variables representing the damage of, e.g., single houses.
The random variables $(\tilde A_i)_{i\ge 1}$, $\tilde A_i \in \N\backslash\{1\}$, describe the number of claims in case of a catastrophe.
The process $N^1$ describes the occurrence of regular claims
whereas a jump of $N^2$ indicates an accumulation of $\tilde A_1$ claims
due to a catastrophe event.

If the insurance company holds the $\xi$-th part of the whole risk, it is exposed to the $\xi$-th part of
the claims.
We model this as
\begin{align}\label{eq:claimsprocess}
C_t^{\xi}=\sum_{i=1}^{N^1_t} Y_{i,1} 1_{\{U_{i,1}\leq\xi_{\tau^1_{i}}\}}+\sum_{i=1}^{N^2_t} Z_i^{\xi}\,, \,\,\text{ where }\,\,
Z_i^{\xi}=\sum_{j=2}^{\tilde A_i}Y_{i,j} 1_{\{U_{i,j}\leq\xi_{\tau^2_{i}}\}}\,.
\end{align}
The random variables $(U_{i,j})_{i,j\ge 1}$ are iid and
$U_{1,1}\sim\cU([0,1])$; they determine whether the company is affected by the
corresponding claim or not.
For fixed $\xi$ this is a thinning of the original process, cf.~\cite[Section 3.12.1]{olofsson}.

We assume independence of $N^1,N^2,(Y_{i,j})_{i,j\ge 1},(\tilde A_i)_{i\ge 1},(U_{i,j})_{i,j\ge 1}$.\\ 

It is possible to write \eqref{eq:claimsprocess0} as a single sum by adapting the jump intensity and the
distribution of the $\tilde A_i$, which we will do to ease the notation in the following.
Note that the jump height distribution does not need to be adapted so that we do not loose the favourable properties for modeling catastrophe events.
Let $L=(L_t)_{t\ge 0}$ be a Poisson process with intensity $\lambda=\lambda^1+\lambda^2$ and jump times $(\tau_{i})_{i\ge 1}$ and let the number of claims per jump of $L$ be denoted by $(A_i)_{i\ge 1}$
with
\begin{align}\label{eq:distrA}
\P(A_1=k)=
\begin{cases}
\frac{\lambda^1}{\lambda^1+\lambda^2} & k=1\,, \\
\frac{\lambda^2}{\lambda^1+\lambda^2}\P(\tilde A_1=k) & k\ge 2 \,.
\end{cases}
\end{align}
We can write the insurance loss index as
\begin{align}\label{eq:claimsprocess0-new}
C_t=\sum_{i=1}^{L_t}Z_i\,, \,\,\text{ where }\,\,Z_i=\sum_{j=1}^{A_i}Y_{i,j}\,.
\end{align}
Both $\lambda$ and the distribution of $A_1$ are chosen such that
\eqref{eq:claimsprocess0} and \eqref{eq:claimsprocess0-new} are equivalent.

The claims process of the insurance company holding the $\xi$-th part of the risk becomes
\begin{align}\label{eq:claimsprocess-new}
C_t^{\xi}=\sum_{i=1}^{L_t}Z_i^{\xi}\,, \,\,\text{ where }\,\,
Z_i^{\xi}=\sum_{j=1}^{A_i}Y_{i,j} 1_{\{U_{i,j}\leq\xi_{\tau_{i}}\}}\,.
\end{align}   

Denote by $a_k:=\P(A_1=k)$.
For all $s\in [0,1]$ the generating function of $A_1$ is given by $G_{A_1}(s):=\sum_{k=1}^\infty a_{k} s^k$, where
\[
G_{A_1}(s)=\frac{\lambda^1}{\lambda^1+\lambda^2}s+\frac{\lambda^2}{\lambda^1+\lambda^2}G_{\tilde A_1}(s)\,.
\]

\begin{assumption}\label{ass:konv-rad}
We assume that
\begin{itemize}\setlength\itemsep{0em}
\item $\E(e^{\eta Y_{1,1}})<\infty$;
\item $\limsup_{k\to\infty}a_{k+1}/a_k<1/\E(e^{\eta Y_{1,1}})$.
\end{itemize}
\end{assumption}

Assumption \ref{ass:konv-rad} implies that the convergence radius of the generating function $G_{A_1}$ is greater than $\E(e^{\eta Y_{1,1}})$ and hence
\begin{align*}
\E&\left(\E(e^{\eta Y_{1,1}})^{A_1}\right)
 =\sum_{k=1}^\infty a_k \E\left(e^{\eta Y_{1,1}} \right)^k
 = G_{A_1}\left( \E\left( e^{\eta Y_{1,1}}\right)\right)
 <\infty\,.
\end{align*}
Note that if $\E\left(e^{\eta Y_{1,1}A_1}\right)<\infty$, also $\E\left(\E(e^{\eta Y_{1,1}})^{A_1}\right)<\infty$ by Jensen's inequality.

In contrast to a model where the claims process is a simple compound Poisson process,
here assuming the existence of exponential moments of the claim size distribution is not a great restriction,
since we model catastrophes as an accumulation of small claims rather than one big claim.\\

The dynamics of the wealth process $X^\xi=(X^\xi_t)_{t\ge0}$ of the insurance company with initial wealth $x$ is given by:
\begin{align}\label{eq:Xdyn}
X^\xi_t
&:= x+\int_0^t q(\xi_s) ds-\sum_{i=1}^{L_t}\sum_{j=0}^{A_i}Y_{i,j} 1_{\{U_{i,j}\leq \xi_{\tau_{i}}\}}
= x+\int_0^t q(\xi_s) ds-C^{\xi}_t\,.
\end{align}

\paragraph{PDMP characterization}

The two-dimensional process $(C,X^\xi)$ is a
PDMP in the sense of \cite{davis1993}. We also refer to \cite[Chapter 8]{BR2011} or \cite{BR2010} for a presentation of the theory.
Our PDMP has the following characteristics:

\vspace{-.5em}
\begin{itemize}\setlength\itemsep{0em}
\item state space $\state$;
\item control space $[0,1]$;
\item deterministic flow $d(C_t,X^\xi_t)=(0, q(\xi_t)) dt$
between jumps;
\item jump intensity $\intens$;
\item jump kernel $Q$, 
\begin{align*}
Q(B|(c,x),\xi)=
\sum_{k=0}^\infty a_k Q_k(B|(c,x),\xi)\,,
\end{align*}
where
\begin{align*}
Q_k(B|(c,x),\xi)=
\sum_{\cK \subseteq \{1,\dots,k\}} \xi ^{|\cK|} \left(1-\xi\right)^{k-|\cK|}
\P\left( \left( \sum_{j=1}^k Y_{1,j},\sum_{j\in \cK} Y_{1,j} \right) \in B-(c,x) \right)\,,
\end{align*}
and where we use the notation $B-(c,x)=\{(b_1-c,b_2-x):(b_1,b_2)\in B\}$;
\item zero running reward rate;
\item zero discount rate. 
\end{itemize}

Denoting by $\tau$ the time of a jump of the PDMP and by $(C_\tau,X_\tau)$ the state immediately after that jump,
we define the set of bounded Markov controls $\cM_b$ as the set of all measurable functions assigning to given input data $(\tau,C_\tau,X_\tau)$ a control until the next jump, i.e.
\[
\stateT\longrightarrow \{\zeta:\R_0^+\longrightarrow [0,1], \zeta \text{ measurable}\}\,.
\]

\section{Utility indifference pricing}
\label{sec:pricing}

The method of utility indifference pricing for the valuation of derivatives
in incomplete markets has been introduced in \cite{Hodges}. It relies on the
fact that even if the derivative cannot be replicated, it may still
be the case that much of its variation can be hedged.

In \cite{Egami} utility indifference pricing is used to price structured catastrophe bonds.
However, there is a difference in modeling the hedging possibility.
In our setup this is done via managing the insured portfolio.  The main idea
is that the loss in the portfolio of a single insurance company is necessarily
correlated with the insurance loss index. The introduction of the derivative has
therefore an influence on the pricing policy of the insurance company.

We will first explain the notion of utility indifference pricing and then apply it to our problem.\\

Assume the investor has a utility function $u$ and
initial wealth $x$. Define
$
J(x,\ell):=\sup_{X_T}\E(u(X_T+ \ell \psi))$, where the supremum is taken over 
all possible wealths $X_T$ that can be generated 
from $x$. The random variable $\psi$ is the payment from a European claim with
expiry $T$, and $\ell$ is the number of claims that are bought.\\

The \emph{utility indifference bid price} $p^b(\ell)$ is the price at which the
investor has the same utility whether she pays nothing and does not receive the claim
$\psi$, or she pays $p^b(\ell)$ now and receives $\ell$ units of the claim $\psi$ at time $T$.
Therefore, $p^b(\ell)$ is the largest amount of money the investor is willing to pay for buying $\ell$ units of the claim $\psi$; it solves
$
J(x-p^b(\ell),\ell)=J(x,0)
$.

The {\em utility indifference ask price} 
$p^a(\ell)$ is the smallest amount of money the investor is willing to accept for selling $\ell$ units of the claim $\psi$; it solves
$
J(x+p^a(\ell),-\ell)=J(x,0)
$.

The two prices are related via $p^b(\ell)=-p^a(-\ell)$.
With this in mind we can define the {\em utility indifference price}
$p:=p^b(1)$.

\begin{assumption}
\begin{itemize}\setlength\itemsep{0em}
\item The insurance company has exponential utility $u(x)=-\exp(-\eta x)$, $\eta>0$.
\item  $X_T$ is of the form $x+\Gamma^\xi_T$ for some control $\xi$ and
$\Gamma^\xi_T$ does not depend on the initial wealth $x$.
\end{itemize}
\end{assumption}

In that case 
\begin{align}
p=-\frac{1}{\eta}\left(\log\left(\inf_{\xi\in\cM_b}\E(\exp(-\eta (\Gamma^{\xi}_T+\psi)))\right)
-\log\left(\inf_{\xi\in\cM_b}\E(\exp(-\eta \Gamma^{\xi}_T))\right)\right)\label{eq:buyers-price}
\end{align}
(provided that the arguments in the logarithms are finite), and hence $p$ does not depend
on the initial wealth $x$.

Note that exponential utility is a natural choice for insurance companies as often such a utility function is used to calculate insurance premia. As an example where exponential utility is used in a stochastic optimal control framework in an insurance context, see \cite{fernandez2008}.

\subsection{The stochastic optimization problem}\label{sec:optimal-strategy}

We apply the concept of utility indifference pricing 
to the model presented in Section \ref{sec:modeling}.
Our aim is to price a derivative written on the total claims process $C$ with payoff
$\psi(C_T)$, where $\psi$ is a continuous and bounded function on $\R_0^+$.

\begin{example}
We are specifically interested in CAT (spread) options, i.e.
$$
\psi(c)=\max(0,\min(c-K,L-K))
$$
with cap $L$ and strike $0<K<L$. The option is in the money, if $c$ exceeds $K$, and the payoff is bounded by $L-K$.

Note that the main task in pricing CAT bonds also lies in pricing the embedded spread option,
since for exponential utility the price of a CAT bond is the sum of a spread option price and a bond price.
\end{example}

We maximize the expected utility from terminal wealth.
The corresponding value function is defined by
\begin{align}\label{eq:value_f}
V(t,c,x):=\sup_{\xi\in \cM_b}\E(u(X^\xi_T+\psi(C_T))|C_t=c,X_t=x)\,.
\end{align}

Since $\psi$ is bounded we have that
for $\xi\equiv 0$,
$\E(u(X^\xi_T+\psi{(C_T)})|C_t=c,X^\xi_t=x)>-\infty$ for all $t,c,x$,
and hence $V(t,c,x)>-\infty$.
$V$ is bounded from above since $u$ is bounded. 
Therefore, $V$ is well-defined.\\

For $v:\stateT\longrightarrow\R$ bounded and measurable the generator of the jump process is defined by
\begin{equation*}
\begin{alignedat}{2}
\cA^\xi v(t,c,x)&= \intens  \sum_{k=1}^\infty a_k  \sum_{\cK \subseteq \{1,\dots,k\}}\xi ^{|\cK|} \left(1-\xi\right)^{k-|\cK|}
 \E\left(v \left( t,c+\sum_{j=1}^k Y_{1,j},x-\sum_{j\in \cK} Y_{1,j} \right) -v(t,c,x) \right)\\
&= \intens  \E\left(v\left(t,c+Z_1,x-Z_1^{\xi}\right)-v(t,c,x)\right)
\,.&
\end{alignedat}
\end{equation*}

The Hamilton-Jacobi-Bellman (HJB) equation corresponding to optimization problem \eqref{eq:value_f} is
\begin{equation}\label{eq:HJB}
\begin{aligned}
v_t(t,c,x)+\sup_{\xi\in[0,1]}\left(q(\xi)v_x(t,c,x)+\cA^\xi v(t,c,x)\right)&=0\,,\\
v(T,c,x)&=u(x+\psi(c))\,.
\end{aligned}
\end{equation}

We make the ansatz $v(t,c,x)=u(x)\exp(-\eta w(t,c))$ to obtain a backward equation which is independent of the initial wealth $x$. This yields
\begin{align*}
v_t(t,c,x)&=-\eta w_t(t,c)v(t,c,x)\,,\\
v_x(t,c,x)&=-\eta v(t,c,x)\,,\\
\tilde{v}(t,c,x,\xi)
&= \intens \E\left(u\left(x-Z_1^{\xi}\right)\exp(-\eta w(t,c+Z_1))-u(x)\exp(-\eta w(t,c))\right)\\
&=v(t,c,x) \intens \E\left(\exp\left(-\eta\left(w(t,c+Z_1)-w(t,c)-Z_1^{\xi}\right)\right)-1\right)\\
&=v(t,c,x) \intens \left(\exp(\eta w(t,c))\E\left(\exp(-\eta w(t,c+Z_1))\exp\left(\eta Z_1^{\xi}\right)\right)-1\right)\,.
\end{align*}
Defining
\begin{align}\label{eq:gen-w}
\tilde{\cA}^\xi w(t,c)&:=-\frac{1}{\eta} \intens \left(\exp(\eta w(t,c))\E\left(\exp(-\eta w(t,c+Z_1))\exp\left(\eta Z_1^{\xi}\right)\right)-1\right)
\end{align}
and using that $v$ is negative,
we obtain the backward equation for $w$:
\begin{equation}\label{eq:HJBw}
\begin{aligned}
w_t(t,c)+\sup_{\xi\in[0,1]}\left(q(\xi)+\tilde{\cA}^\xi w(t,c)\right)&=0\,,\\
w(T,c)&=\psi(c)\,.
\end{aligned}
\end{equation} 

\begin{lemma}\label{lem:w-bounded}
Let $W$ be such that $V(t,c,x)=u(x)\exp(-\eta W(t,c,x))$.
Then $W$ is bounded by $\|q\|_{\infty}  T+\|\psi\|_\infty$.
\end{lemma}

\begin{proof}
We have $V(t,c,x)=u(x)\exp(-\eta W(t,c,x))$, i.e.
\begin{align*}
W(t,c,x)&=-\frac{1}{\eta}\log\left(\frac{V(t,c,x)}{u(x)}\right)
=-\frac{1}{\eta}\log\left(\inf_{\xi \in\cM_b}\E\left(\frac{u(X^\xi_T+\psi(C_T))}{u(x)}\Big|C_t=c,X^\xi_t=x\right)\right)\\
&=-\frac{1}{\eta}\log\left(\inf_{\xi \in\cM_b}\E\left(\exp(-\eta(X^\xi_T-x+\psi(C_T)))\Big|C_t=c,X^\xi_t=x\right)\right)\,.
\end{align*}
Denote by $X^0$ the process $X^\xi$ with $\xi\equiv 0$. Then
\begin{align*}
\inf_{\xi \in\cM_b}&\E\left(\exp(-\eta(X^\xi_T-x+\psi(C_T)))\Big|C_t=c,X^\xi_t=x\right)
\le \E\left(\exp(-\eta(X^0_T-x+\psi(C_T)))\Big|C_t=c,X^0_t=x\right)\\
&= \E\left(\exp(-\eta(q(0)(T-t)+\psi(C_T)))\Big|C_t=c,X^0_t=x\right)
\le \exp(\eta(\|q\|_{\infty} (T-t)+\|\psi\|_\infty))\,,
\end{align*}
and
\begin{align*}
\inf_{\xi \in\cM_b}&\E\left(\exp(-\eta(X^\xi_T-x+\psi(C_T)))\Big|C_t=c,X^\xi_t=x\right)\\
&\ge \inf_{\xi \in\cM_b}\E\left(\exp(-\eta(\|q\|_{\infty} (T-t)-C^\xi_T+\psi(C_T)))\Big|C_t=c\right)
\ge \exp(-\eta(\|q\|_{\infty} (T-t)+\|\psi\|_\infty))\,.
\end{align*}
Thus $|W(t,c,x)|\le \|q\|_{\infty}  T+\|\psi\|_\infty$.
\end{proof}


\subsection{Verification result}
\label{subsec:Verification}

We show that the solution of the HJB equation \eqref{eq:HJB} solves the optimization problem \eqref{eq:value_f}. 
For this we apply results from stochastic control theory for PDMPs;
more precisely, a slight variation of the verification theorem 
\cite[Theorem 8.2.8]{BR2011}.
For this we recall two definitions from \cite{BR2010}. 

\begin{definition}\label{def:bounding}
A measurable function $b:\state\longrightarrow \R_0^+$ is called a
{\em bounding function} for our piecewise deterministic Markov decision model,
if there exist constants $c_u,c_Q,c_{\text{flow}}\ge 0$ such that
for all $(c,x)\in \state$
\begin{enumerate}
\renewcommand{\theenumi}{\roman{enumi}}
\item $|u(x+\psi(c))|\le c_u b(c,x)$;
\item $\int b(\tilde c,\tilde x)Q(d\tilde c\times d\tilde x|(c,x),\xi)\le c_Q b(c,x)$ for all $(c,x)\in \state$, $\xi\in [0,1]$;
\item $b(c,x+\int_0^T \int_0^1 q(\xi)r_s(d\xi)ds)\le c_{\text{flow}}b(c,x)$ for all $r\in \cR$.
\end{enumerate}
Here $\cR$ is the space of {\em relaxed policies}, i.e.~of measurable maps
$\R_0^+\longrightarrow  \cP([0,1])$, where $\cP([0,1])$ is the space
of all probability measures on the Borel $\sigma$-algebra on $[0,1]$.
\end{definition}

\begin{definition}\label{def:gamma}
Let $b:\R_0^+\times \R\longrightarrow \R_0^+$ be a bounding function and $\gamma>0$ fixed.
Define the Banach space
$
\B_{b,\gamma}:=\{v:\stateT\longrightarrow \R_0^+: v \text{ measurable and }\|v\|_b<\infty\}\,,
$
with the norm
\[
\|v\|_{b,\gamma}:=\underset{(t,c,x)}{\mathrm{ess\;sup}}\frac{|v(t,c,x)|}{\exp(\gamma(T-t))b(c,x)}\,, \,\,\text{ where }\,\,\frac{0}{0}:=0\,.
\]
\end{definition}

\begin{theorem}[Verification Theorem]\label{th:verification}
Let $b$ be a bounding function for our piecewise deterministic Markov decision model with 
$\E\big(|b(C_T,X^\xi_T)|\big|C_t=c,X_t=x\big)<\infty$ for all 
$\xi,t,c,x$. Let 
$v\in \cC^{1,0,1}(\stateT)\cap \B_{b,\gamma}$
be a solution of the HJB equation \eqref{eq:HJB} and let $\alpha^\ast$ be a 
maximizer
for \eqref{eq:HJB}, leading to the state process $\left(C,X^{\xi^\ast}\right)$.

Then $v=V$ and $\xi^\ast=\alpha^\ast(t,C_{t-},X^{\xi^\ast}_{t-})$ is an optimal feedback-type Markov policy.
\end{theorem}

\begin{remark}
In the statement of \cite[Theorem 8.2.8]{BR2011} there is 
another condition required, namely that $\alpha_b<1$ for a
constant $\alpha_b$ depending on $b,Q$ and the arbitrary $\gamma$ from
Definition \ref{def:gamma}. But it is argued in \cite{BR2010}  that
for finite horizon problems $\gamma$ can always be chosen large enough to 
satisfy $\alpha_b<1$.
\end{remark}

For proving Theorem \ref{th:verification}, we first need to prove existence of a bounding function.

\begin{lemma}\label{lem:bounding}
The function $b$ defined by $b(c,x):=\exp(\eta |x|)$ is a bounding function for our piecewise deterministic Markov decision model.
\end{lemma}

\begin{proof}
We need to check the conditions given in Definition \ref{def:bounding}.
\begin{enumerate}
\renewcommand{\theenumi}{\roman{enumi}}
\item $u(x+\psi(c))=-\exp(-\eta(x+\psi(c))$ such that 
$|u(x+\psi(c))|=\exp(-\eta(x+\psi(c))\le \exp(\eta \|\psi\|_\infty)b(c,x)$.
\item
$
\int  b(\tilde c,\tilde x)Q(d\tilde c\times d\tilde x|(c,x),\xi)
=\int \exp(\eta |\tilde x|)Q(d\tilde c\times d\tilde x|(c,x),\xi)
=\E\left(\exp\left(\eta \Big|x-Z_1^{\xi}\Big|\right)\right)\\
\le b(c,x)\E\left(\exp\left(\eta Z_1^{\xi}\right)\right)
$.
\item 
$
b\left(c,x+\int_0^T\int_0^1 q(\xi)r_s(d\xi)ds\right)=\exp\left(\eta \Big |x+\int_0^T\int_0^1 q(\xi)r_s(d\xi)ds\Big |\right)
\le\exp(\eta \|q\|_{\infty}  T)b(c,x)
$.
\end{enumerate}
\end{proof}

Now we need to show that the backward equation \eqref{eq:HJBw} has a solution and hence also \eqref{eq:HJB} has a solution.\\

Define $\cH$ on $\cC_b(\R_0^+)$ by
\begin{align}\label{eq:tildeH}
 (\cH \varphi)(c):=\sup_{\xi\in[0,1]}\left(q(\xi)+(\tilde{\cA}^\xi \varphi)(c)\right)\,.
\end{align}
We show that if $\varphi\in \cC_b(\R_0^+)$, then $\cH \varphi \in \cC_b(\R_0^+)$ and that $\cH$ is locally Lipschitz.
For this we write $\cH=g \circ h \circ f$ and show that $g,h,f$ are locally Lipschitz and $g$ is $\cC_b(\R_0^+)$-valued.

\begin{lemma}\label{lem:power-series}
For $\sigma \in \cC_b(R_0^+)$ the mapping $\xi \mapsto \E(\sigma(Z_1)\exp(\eta Z_1^{\xi}) )+\intens/\eta$
is a power series in $\xi$.
Its coefficients are of the form 
$h_k(\sigma)=\E(\delta_k \sigma(Z_1))$ for non-negative random variables $\delta_k$ with $\sum_{k=0}^\infty\E(\delta_k)<\infty$ that do not dependent on $\xi$ and $\sigma$. 
The power series converges uniformly on $[0,1]$. 
\end{lemma}

\begin{proof} 
Let $F$ be the distribution function of $Y_{1,1}$. Then

\begin{align*}
\E&\left(\sigma(Z_1)\exp(\eta Z_1^{\xi}) \right)
=\sum_{k=1}^\infty a_k \E\left(\sigma(Z_1)\exp(\eta Z_1^{\xi}) \Big|A_1=k\right)\\
&=\sum_{k=1}^\infty a_k\int\dots\int \sigma\left(\sum_{j=1}^k y_j\right)\E\left(\exp \left(\eta \sum_{j=1}^k y_j 1_{\{U_{1,j}\leq \xi\}}\right) \right)d F(y_1)\dots d F(y_k)\\
&=\sum_{k=1}^\infty a_k\int\dots\int \sigma\left(\sum_{j=1}^k y_j\right)\prod_{j=1}^k\E\left(\exp \left(\eta  y_j 1_{\{U_{1,j}\leq \xi\}}\right) \right)d F(y_1)\dots d F(y_k)\\
&=\sum_{k=1}^\infty a_k\int\dots\int \sigma\left(\sum_{j=1}^k y_j\right)\prod_{j=1}^k (\xi \exp(\eta y_j)+(1-\xi)) d F(y_1)\dots d F(y_k)\\
&=\sum_{k=1}^\infty a_k\E \left( \sigma(Z_1)\prod_{j=1}^k (\xi (\exp(\eta Y_{1,j})-1)+1) \Big|A_1=k\right)\,.
\end{align*}
Expanding the above expression yields the first and the second claim of the lemma.
Setting $\sigma\equiv1$, we get $\sum_{k=0}^\infty\E(\delta_k)<\infty$.
Setting $\xi=1$ gives
\begin{align*}
\left|\sum_{k=1}^\infty a_k\E \left( \sigma(Z_1)\prod_{j=1}^k  \exp(\eta Y_{1,j}) \Big|A_1=k\right)\right|
\le \|\sigma\|_\infty \sum_{k=1}^\infty a_k\E \left( \exp\left(\eta \sum_{j=1}^k Y_{1,j}\right) \Big|A_1=k\right)\,.
\end{align*}
The right-hand side is finite by Assumption \ref{ass:konv-rad} and by the assumption of the lemma.
\end{proof}

Define the function-space
\begin{align*}
\Lambda:=\left\{\phi:\N_0  \longrightarrow \cC_b(\R_0^+) \colon \|\phi\|_{\Lambda}:=\sum_{k=0}^{\infty}\|\phi_k\|_{\infty}<\infty \right\}
\end{align*}
and let $h:\cC_b(\R_0^+ \times \R_0^+)\longrightarrow\Lambda$ be defined by
$h_k(\phi)(c)=\E(\delta_k \phi(c,Z_1))$, with $\delta_k$ 
as in Lemma \ref{lem:power-series}. Hence, for 
any $\phi\in \cC_b(\R_0^+ \times \R_0^+)$,
$\E\left(\phi(c,Z_1)\exp(\eta Z_1^\xi)\right)=\sum_{k=0}^\infty h_k(\phi)(c) \xi^k$
for every $c\in \R_0^+$.\\

\begin{lemma}\label{lem:h}
The function $h$ is a bounded linear operator.
\end{lemma}

\begin{proof}
We need to prove that for every $k\in \N_0$ the mapping $h_k$ is a bounded linear operator 
$\cC_b(\R_0^+ \times \R_0^+)\longrightarrow \cC_b(\R_0^+)$.

Let $\phi\in \cC_b(\R_0^+ \times \R_0^+)$. We show that the mapping
$c\mapsto \E(\delta_k \phi(c,Z_1))$ is continuous and bounded on $\R_0^+$.
Let $c_n \to c$ in $\R_0^+$. 
Then $\phi(c_n,z) \to \phi(c,z)$ for all 
$z\in \R_0^+$.
 The sequence $(\phi(c_n,\cdot)\delta_k)_{n\ge 0}$ is dominated  by 
$\|\phi\|_\infty \delta_k$, which is integrable.
Hence, $\E(\phi(c_n,Z_1)\delta_k)\to\E(\phi(c,Z_1)\delta_k)$ by the dominated convergence theorem. Thus $h_k(\phi)$ is continuous. Moreover, $h_k(\phi)$ is
bounded, since $\left|\E(\phi(c,Z_1)\delta_k)\right|\le \|\phi\|_\infty \E(\delta_k)$.

For $\phi\in \cC_b(\R_0^+ \times \R_0^+)$ it holds that
 $\|h\|_\Lambda\le \sum_{k=0}^\infty\|h_k(\phi)\|_\infty =\sum_{k=0}^\infty \sup_c|\E(\delta_k\phi(c,Z_1))| \le \|\phi\|_\infty\sum_{k=0}^\infty\E(\delta_k)$.
Thus $h$ is bounded by Lemma \ref{lem:power-series}.
 \end{proof}

For a sequence $(x_k)_{k\ge 0}$ in $\R$ define the function $\tilde g$ by
\begin{align}\label{eq:tildeg}
\tilde g(x):=\sup_{\xi\in[0,1]}\left(q(\xi)+\frac{\intens}{\eta}-\sum_{k=0}^\infty x_k \xi ^k\right)\,.
\end{align}

\begin{lemma}
\label{lem:g-loc-lip}
Let $\tilde g$ be defined as in \eqref{eq:tildeg}. Then
\begin{enumerate}
\item $\tilde g$ is defined on $\ell^1$ and it is bounded on every norm-bounded subset of $\ell^1$;
\item $\tilde g$ is convex;
\item $\tilde g$ is Lipschitz on every norm-bounded subset of $\ell^1$.
\end{enumerate}
\end{lemma}

\begin{proof}
For $x=(x_k)_{k\ge 0}\in\ell^1$ and $\xi\in [0,1]$ we have $|\sum_{k=0}^\infty x_k \xi^k|\le \|x\|_1$. The function $\xi\mapsto  \sum_{k=0}^\infty x_k \xi^k$, $\xi\in [0,1]$ is well-defined and continuous as a uniform limit of continuous 
functions on $[0,1]$. 
Since $q$ is also continuous, the first statement follows.
 
The proof of the second statement is straightforward.

Following \cite{robvar74} we use the convexity of $\tilde g$ to show that $\tilde g$ is Lipschitz on $\{x\in\ell^1:\|x\|\le R\}$.
Let $\|x\|,\|y\|\le R$ and define $z:=y+\frac{R}{\|y-x\|}(y-x)$.
It holds that $\|z-y\|=R$ and hence $\|z\|\le2R$.
By the definition of $z$ we have that $y=\beta z + (1-\beta)x$, where $\beta=\|y-x\|/(\|y-x\|+R)$.
Since $\tilde g$ is convex, $\tilde g(y)\le \beta \tilde g(z)+(1-\beta)\tilde g(x)$ and hence
$\|\tilde g(y)-\tilde g(x)\| = \|\beta(\tilde g(z)-\tilde g(x))\| \le 2\beta \sup_{\|z\|\le2R}|\tilde g(z)| \le 2\beta c \le (2C/R) \|y-x\|$ for some constant $c>0$, since $\tilde g$ is bounded on $\{z\in\ell^1:\|z\|\le2R\}$.
\end{proof}

For $\phi \in \Lambda$ let $\phi(c):=(\phi_k(c))_{k\ge 0}$ and define the function $g$ by $g(\phi)(c)=\tilde g(\phi(c))$.

\begin{lemma}\label{lem:g}
 The function $g$ is $\cC_b(\R_0^+)$-valued and locally Lipschitz.
\end{lemma}

\begin{proof} Let $\phi=(\phi_k)_{k\ge 0}\in \Lambda$.
 Let $c_n\to c$ in $\R_0^+$ and let $\varepsilon >0$. There exists $k_0\in \N_0$ such that $\sum_{k\ge k_0} \|\phi_k\|_\infty<\varepsilon/4$, and for $n$ large enough $\sum_{k=0}^{k_0} |\phi_k(c_n)-\phi_k(c)|<\varepsilon/2$.
 Thus, $\sum_{k=0}^\infty |\phi_k(c_n)-\phi_k(c)| \le \sum_{k=0}^{k_0} |\phi_k(c_n)-\phi_k(c)|+\sum_{k=k_0+1}^{\infty} |\phi_k(c_n)-\phi_k(c)|< \varepsilon$.
 
 The claim that $g$ is locally Lipschitz follows from Lemma \ref{lem:g-loc-lip}:
 let $R>0$ and let  $\phi^1,\phi^2\in\Lambda$ with $\|\phi^1\|_\Lambda\le R$
and $\|\phi^2\|_\Lambda\le R$. $\tilde g$ is Lipschitz on the ball with radius
$R$ in $\ell^1$. Denote the corresponding Lipschitz constant by $L_R$. 
Then $\phi^1(c),\phi^2(c)$ lie in the ball with radius $R$ in $\ell^1$.
Hence, $\|g(\phi^1)-g(\phi^2)\|_\infty=\sup_c|\tilde g(\phi^1(c))-\tilde g(\phi^2(c))|\le L_R \|\phi^1(c)-\phi^2(c)\|_1 \le L_R\|\phi^1-\phi^1\|_\Lambda$.
\end{proof}

Finally, define $f:\cC_b(\R_0^+) \longrightarrow \cC_b(\R_0^+ \times \R^+)$, $f(w)(c,z):=\exp(-\eta(w(c+z)-w(c))$
and note that $f$ is locally Lipschitz.

\begin{lemma}\label{lem:H-loc-lip}
 Let $\cH$ be defined as in \eqref{eq:tildeH}.
 If $\varphi\in \cC_b(\R_0^+)$, then $\cH \varphi \in \cC_b(\R_0^+)$ and $\cH$ is locally Lipschitz.
\end{lemma}

\begin{proof}
 We have $\cH=g \circ h \circ f$. The first claim follows from Lemma \ref{lem:g}.
 Further, $\cH$ is locally Lipschitz as a concatenation of locally Lipschitz functions.
 The latter follows from Lemma \ref{lem:h} and Lemma \ref{lem:g}.
\end{proof}

Now we prove that \eqref{eq:HJBw} has a unique maximal local solution.

\begin{lemma}\label{lem:backward_ode}
Let $\psi \in \cC_b(\R)$.
Then the backward equation \eqref{eq:HJBw} has a unique maximal local solution.
\end{lemma}

\begin{proof}
The backward equation \eqref{eq:HJBw} is an initial value problem with $\cC_b(\R_0^+)$-valued solution:
\begin{align}\label{eq:ode}
\varphi'(t)=-\cH\varphi(t)\,, \,\,\text{ and }\,\,
\varphi(T)=\psi(c)\,.
\end{align}

By Lemma \ref{lem:H-loc-lip}, $\cH$ is $\cC_b(\R_0^+)$-valued and locally
Lipschitz. In particular, $\cH$ is Lipschitz on the ball with
radius $2(\|q\|_{\infty} T+\|\psi\|_\infty)$.
From the Picard-Lindel\"of theorem
on existence and uniqueness of solutions of ordinary differential equations we
get existence and uniqueness of a maximal local solution of \eqref{eq:ode},
i.e.~there exists $\varepsilon>0$ and a solution $\varphi$ 
of \eqref{eq:ode} on $[T-\varepsilon,T]$ with $\|\varphi(t)\|_\infty\le 2 (\|q\|_{\infty}  T+\|\psi\|_\infty)$ for all $t\in [T-\varepsilon,T]$. 
We may choose $\varepsilon$ maximal such that $\varepsilon=T$ or $\|\varphi(T-\varepsilon)\|_\infty=2 (\|q\|_{\infty}  T+\|\psi\|_\infty)$.

The function $w:[T-\varepsilon, T]\times \R_0^+\longrightarrow \R$ 
defined by $w(t,c)=\varphi(t)(c)$ is the unique maximal local solution of \eqref{eq:HJBw}.
\end{proof}

\begin{proof}[Proof of Theorem \ref{th:verification}]
Since for every $\varphi\in \cC_b(\R_0^+)$ and $c\in \R_0^+$ the function 
$\xi\mapsto q(\xi) +(\tilde {\cal A}^\xi w)(c)$ is continuous 
on $[0,1]$ by Lemma \ref{lem:power-series}, there exists a maximizer for \eqref{eq:HJB}.
By Lemma \ref{lem:bounding} and Lemma \ref{lem:backward_ode} the assumptions of
Theorem \ref{th:verification} are satisfied on
$[T-\varepsilon,T]\times\state\longrightarrow\R$, where $\varepsilon$ is as in the
proof of Lemma \ref{lem:backward_ode}.  Along the lines of the proof of
\cite[Theorem 8.2.8]{BR2011} it can be shown that
$v:[T-\varepsilon,T]\times\state\longrightarrow\R$ with $v(t,c,x)=u(x)\exp(-\eta w(t,c))$ solves the optimization problem \eqref{eq:value_f} for $t\in[T-\varepsilon,T]$, i.e.~$v(t,\cdot,\cdot)=V(t,\cdot,\cdot)$ for $t\in[T-\varepsilon,T]$.

By Lemma \ref{lem:w-bounded} it holds that
$|w(t,c)|=|-1/\eta\log(v(t,c,x)/u(x))|
=|-1/\eta\log(V(t,c,x)/u(x))|=|W(t,c,x)|\le \|q\|_{\infty}  T+\|\psi\|_\infty$
for $t\in [T-\varepsilon,T]$.
Therefore, $\|w(T-\varepsilon,.)\|_\infty<2( \|q\|_{\infty}  T+\|\psi\|_\infty)$ and hence
$\varepsilon=T$.
Thus $w$ solves \eqref{eq:HJBw} on the whole of $[0,T]\times \R_0^+$ and therefore
 $v$ solves the optimization problem \eqref{eq:value_f} in the whole of
$\stateT$.  
\end{proof}

\subsection{Utility indifference price}
\label{subsec:uip}

With the solution $w$ of \eqref{eq:HJBw} we can compute the utility indifference price $p$ of a derivative with payoff $\psi$. Given the value of the index $c$ and the amount of wealth $x$ at time $t$, the maximum expected utility of terminal wealth can be written as
\begin{align*}
V(t,c,x-p(t,c,x))&=u(x-p(t,c,x))\exp(-\eta w(t,c))=u(x)\exp(\eta p(t,c,x))\exp(-\eta w(t,c))\,.
\end{align*}
The corresponding value with no derivative bought is given by
\begin{align*}
V^0(t,c,x)=u(x)\exp(-\eta w^0(t,c))\,,
\end{align*}
where $w^0$ is the solution of equation \eqref{eq:HJBw} with $\psi\equiv 0$.
Therefore \eqref{eq:buyers-price} simplifies to
\begin{align} \label{eq:indiff_price}
p(t,c,x)=w(t,c)-w^0(t,c)\,.
\end{align}
In particular, $p$ does not depend on $x$ and we omit that parameter from $p$ 
henceforth.
The function $w^0$ does not depend on $c$. So $w^0$ is the solution of an ordinary differential equation.

\begin{lemma}
Let $w^0$ solve \eqref{eq:HJBw} with terminal condition $w^0(T,c)\equiv0$.
Then
\begin{align*}
w^0(t,c)=(T-t)\sup_{\xi\in[0,1]}\left( q(\xi)+\frac{\intens}{\eta} \left(1-G_{A_1}\left( \xi\left(\E\left(\exp\left(\eta Y_{1,1}\right)\right)-1\right)+1 \right) \right)\right)\,.
\end{align*}
\end{lemma}

\begin{proof}
Since $w^0$ does not depend on $c$, the backward equation \eqref{eq:HJBw} becomes
\begin{equation}
\begin{aligned}\label{eq:HJBw0}
w^0_t(t,c)+\sup_{\xi\in[0,1]}\left( q(\xi)+\frac{\intens}{\eta}\left(1-\E\left( \exp\left(\eta Z_1^{\xi}\right)\right)\right)\right)&=0\,,\\
w^0(T,c)&=0\,.
\end{aligned}
\end{equation}
We calculate the expected value in \eqref{eq:HJBw0}:
\begin{align*}
\E&\left(\exp\left(\eta Z_1^{\xi}\right)\right)=
\sum_{k=1}^\infty a_k \E\left(\exp\left(\eta Z_1^{\xi}\right)\Big|A_1=k\right)\\
&= \sum_{k=1}^\infty a_k \E\left(\prod_{j=1}^{k}\left(\xi(\exp\left(\eta Y_{1,1}\right)-1)+1 \right)\right)
=\sum_{k=1}^\infty a_{k} \left( \xi\left(\E\left(\exp\left(\eta Y_{1,1}\right)\right)-1\right)+1 \right)^{k}\\
&=G_{A_1}\left( \xi\left(\E\left(\exp\left(\eta Y_{1,1}\right)\right)-1\right)+1 \right)\,.\end{align*}
Integration yields the claimed solution.
\end{proof}

From \eqref{eq:indiff_price} we can derive a backward equation for $p$. Since $w^0$ does not depend
on the second variable, we have
\begin{align*}
\tilde{\cA}^\xi& w(t,c)
=-\frac{1}{\eta} \intens \left(\exp(\eta w(t,c))\E\left(\exp(-\eta w(t,c+Z_1))\exp\left(\eta Z_1^{\xi}\right)\right)-1\right)\\
&=-\frac{1}{\eta} \intens \left(\exp\big(\eta (w(t,c)-w^0(t,c))\big)\E\left(\exp\big(-\eta (w(t,c+Z_1)-w^0(t,c))\big)\exp\left(\eta Z_1^{\xi}\right)\right)-1\right)\\
&=-\frac{1}{\eta} \intens \left(\exp\big(\eta (w(t,c)-w^0(t,c))\big)\E\left(\exp\big(-\eta (w(t,c+Z_1)-w^0(t,c+Z_1))\big)\exp\left(\eta Z_1^{\xi}\right)\right)-1\right)\\
&=\tilde{\cA}^\xi p(t,c)
\end{align*}
and it holds that $p_t(t,c)=w_t(t,c)-w^0_t(t,c)=w_t(t,c)-\bar w$, where 
\begin{align*}
\bar w=\sup_{\xi\in[0,1]}\left( q(\xi)+\frac{\intens}{\eta} \left(1-G_{A_1}\left( \xi\left(\E\left(\exp\left(\eta Y_{1,1}\right)\right)-1\right)+1 \right) \right)\right)\,.
\end{align*}
Hence
\begin{align}\label{eq:backward-p}
p_t(t,c)
=\bar w-\sup_{\xi\in[0,1]}\left(q(\xi)+\tilde{\cA}^\xi w(t,c)\right)
=\bar w-\sup_{\xi\in[0,1]}\left(q(\xi)+\tilde{\cA}^\xi p(t,c)\right)\,.
\end{align}

\section{Computations}
\label{sec:numerics}

In this section we present a convenient numerical method for computing the expected value in \eqref{eq:backward-p} for the case where the distribution of $Y_{1,1}$ has a smooth density.
Denote by $f^{\ast k}$ the k-fold convolution of a function $f$ with itself, i.e.
$f^{\ast 2}=f\ast f$ and $f^{\ast (k+1)}=f\ast f^{\ast k}$. Let $\hat F$ denote 
the Fourier transform. It holds that $\hat F(f^{\ast k})=\hat F(f)^k$.
The following lemma gives an efficient method for computing $\E(\exp(-\eta w(t,c+Z_1))\exp(\eta Z_1^{\xi}))$.

\begin{lemma}
\label{lemma:fourier}
Assume that the distribution of $Y_{1,1}$ has a piecewise continuous  density
$\mu$. Denote by $\tilde \mu(z)=\exp(\eta z) \mu(z)$. 
Let $\sigma$ be measurable and bounded.

Then it holds that
\begin{align*}
\E \left(\sigma(Z_1)e^{\eta Z_1^{\xi}}\right)=
\int_0^{\infty}\sigma(z)\mu_\xi(z) dz\,,
\end{align*}
where $\mu_\xi=\hat F^{-1}\left(G_{A_1}(\hat F(\xi \tilde \mu+(1-\xi)\mu))\right)$.
\end{lemma}

\begin{proof}
Denote $\bar Y_{k_1,k_2}:=\sum_{j=k_1}^{k_2} Y_{1,j}$. We have
\begin{align*}
\E &\left(\sigma(Z_1)e^{\eta Z_1^{\xi}}\Big|A_1=k\right)
=\E\left(\sigma(\bar Y_{1,k})\prod_{j=1}^k\left( \xi e^{\eta Y_{1,j}}+(1-\xi)\right)\right)\\
&=\sum_{j=0}^{k}\binom{k}{j}\xi^j(1-\xi)^{k-j}\E(\sigma(\bar Y_{1,k})e^{\eta \bar Y_{1,j}})\\
&=\sum_{j=0}^{k}\binom{k}{j}\xi^j(1-\xi)^{k-j}\int_0^{\infty}\int_0^{\infty}\sigma(
\bar y_{1,j} + \bar y_{j+1,k})e^{\eta \bar y_{1,j}}\mu^{*j}(\bar y_{1,j})\mu^{*(k-j)}(\bar y_{j+1,k})d \bar y_{1,j} d \bar y_{j+1,k}\\
&=\sum_{j=0}^{k}\binom{k}{j}\xi^j(1-\xi)^{k-j}\int_0^{\infty}\sigma(\bar y_{1,k})\int_0^{\infty}e^{\eta (\bar y_{1,k}-\bar y_{j+1,k})}\mu^{*j}(\bar y_{1,k}-\bar y_{j+1,k})\mu^{*(k-j)}(\bar y_{j+1,k})d \bar y_{j+1,k} d \bar y_{1,k}\\
&=\sum_{j=0}^{k}\binom{k}{j}\xi^j(1-\xi)^{k-j}\int_0^{\infty}\sigma(\bar y_{1,k})\int_0^{\infty}\tilde \mu^{*j}(\bar y_{1,k}-\bar y_{j+1,k})\mu^{*(k-j)}(\bar y_{j+1,k}) d \bar y_{j+1,k} d \bar y_{1,k}\,,
\end{align*}
where we used that $\tilde \mu^{\ast k}(z)=\exp(\eta z) \mu^{\ast k}(z)$,
which can be seen by induction. The first claim now follows by linearity of the
convolution. Further,
\begin{align*}
\E \left(\sigma(Z_1)e^{\eta Z_1^{\xi}}\right)
&=\sum_{k=1}^{\infty}a_k\E\left(\sigma(Z_1)e^{\eta Z_1^{\xi}}\Big|A_1=k\right)
=\sum_{k=1}^{\infty}a_k\int_0^{\infty}\sigma(z)\left(\xi \tilde \mu+(1-\xi)\mu\right)^{\ast k}(z) dz\\
&=\int_0^{\infty}\sigma(z)\sum_{k=1}^{\infty}a_k\left(\xi \tilde \mu+(1-\xi)\mu\right)^{\ast k}(z) dz\,.
\end{align*}
With this,
\begin{align*}
\hat F\left(\sum_{k=1}^{\infty}a_k\left(\xi \tilde \mu+(1-\xi)\mu\right)^{\ast k}\right)
&=\sum_{k=1}^{\infty}a_k\hat F\left(\left(\xi \tilde \mu+(1-\xi)\mu\right)^{\ast  k}\right)\\
=\sum_{k=1}^{\infty}a_k\left(\hat F\left(\xi \tilde \mu+(1-\xi)\mu\right)\right)^{k}
&=G_{A_1}\left(\hat F(\xi \tilde \mu+(1-\xi)\mu)\right)\,.
\end{align*}

Since $(\xi \tilde \mu+(1-\xi)\mu)$ is piecewise continuous,
$(\xi \tilde \mu+(1-\xi)\mu))^{*k}$ is continuous for $k\ge 2$ and hence
$\sum_{k=1} a_k\left(\xi \tilde
\mu+(1-\xi)\mu\right)^{*k}$ is piecewise continuous.

Thus $\sum_{k=1}^\infty a_k\left(\xi \tilde
\mu+(1-\xi)\mu\right)^{*k}=\hat F^{-1}\left(\hat F\left(\sum_{k=1}^\infty a_k\left(\xi \tilde
\mu+(1-\xi)\mu\right)^{*k}\right)\right)$ a.e.~on $\R_0^+$.
\end{proof}

\subsection{Numerical experiments}
\label{subsec:exp}

Our aim is to price a CAT (spread) option, i.e.
$\psi(C_T)=\max(0,\min(C_T-K,L-K))$.

The function $q$ governing the company's market share is chosen
as in \cite{LeoNgare}.
We assume that there are
$M$ clients in the market who potentially contribute to the claims
process. Let $a$ be the \emph{fair} annual premium 
for one client, i.e.~$\E(C_1)=M a$.
The annual premium for one contract
therefore has to be
greater or equal than $a$, since otherwise the insurance company will make an
almost sure loss in the long run.
The premium the insurance company charges for a claim is $a(1+\theta)$ with $\theta>0$.
Furthermore, the company faces an exogenously given demand curve $d$ for insurance.
It is continuous, decreasing in $\theta$, and satisfies
$d(\theta)=M$ for $\theta\le0$ and 
$d(\theta)=0$ for $\theta\ge m$, i.e.~the company gets 
to insure the whole risk, if it does not charge any risk loading ($\theta=0$) and it gets $0$ contracts, if
the risk-loading exceeds some fixed number $m>0$. 
With this $q(\xi)=\xi a (1+\theta(\xi))$, where $\theta(\xi)=d^{-1}(\xi M)$.
For our numerical example we choose
\[
d(\theta) =\begin{cases}
M & \theta \le 0 \\
M(1-\theta/m) &  0<\theta<m\\
0 &  \theta \ge m\,.
\end{cases}
\]
The model parameters in our example are given
by  $M=10^4$, $m=2$, $K=10^7$,  $\eta=10^{-6}$, $T=1$ year. 
The number of jumps in case of a catastrophe is Poisson distributed, $(\tilde A_1-2) \sim \text{Poisson}(40)$; the distribution of $A_1$ then follows from \eqref{eq:distrA}.
The claim size distribution is a Gamma distribution, $Y_{1,1}\sim \text{Gamma}(10,5000)$.\\

We want to study two effects: the effect that holding a derivative has on the risk loading (which depends on the optimal market share $\xi^\ast$) and its change over time, and the effect of our model (the \emph{clustered claims (CC) model}) on the utility indifference price of the derivative and on the risk loading in comparison to the model where the claims process is a simple compound Poisson process as, e.g., in \cite{LeoNgare} (the \emph{single claim (SC) model}).

For the SC-model we choose $\lambda_1=100, \lambda_2=0$.
In order to be able to compare the two models we adapt $\lambda_1,\lambda_2$ such that the expected annual claim size per contract $a$ stays constant, yielding $\lambda_1=69, \lambda_2=1$.\\

Figure \ref{fig:bidprice} shows the utility indifference price $p$ of the CAT spread option in dependence of the value $c$ of the claims process.
We see that the price increases in $c$.
As time increases the expected number of claims within the remaining time decreases, and hence also the price decreases;
for $t\to T$ the prices converges to the payoff.
Further, we observe that in the SC-model the price is always lower than in the
CC-model, since the latter more accurately accounts for a clustering of claims.\\

\begin{figure}
\centering
\includegraphics[height=.25\textwidth]{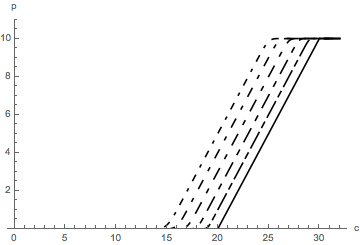}
\includegraphics[height=.25\textwidth]{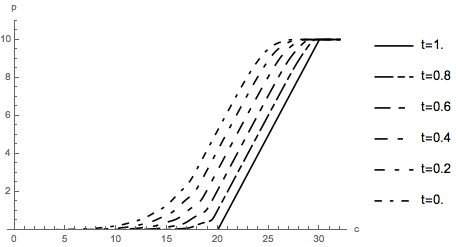}
\caption{Utility indifference prices $p$  in the SC-model (left) and in the CC-model (right). (Units on the axes are $10^6$ units of currency.)}
\label{fig:bidprice}
\end{figure}

Figure \ref{fig:control} shows the risk loading corresponding to the optimal market share $\xi^\ast$ in dependence of the value $c$ of the claims process.
For small $c$ the risk loading is the same as in the case of no derivative held,
since the probability that the derivative has a positive payoff is small.
For $c>L$ a further increase of $c$ does not change the payoff and hence the situation is the same as for holding no derivative.

As time increases the probability that the payoff of the derivative grows in $c$ and hence compensates losses during the remaining time decreases, but also the expected number of claims before $T$ decreases.
An interesting observation is that for small $c$ the first effect dominates and hence the risk loading increases,
whereas for large $c$ the latter effect dominates and hence the risk loading decreases.

In the CC-model the risk loading is in general higher than in the SC-model, since we imposed risk aversion.
The risk loading decreases significantly when a derivative is bought.
The effect of holding a derivative is higher in the CC-model; the optimal average risk loading decreases by approximately $31.9\%$ (compared to $2.7\%$ in the SC-model).\\

\begin{figure}
\centering
\includegraphics[height=.25\textwidth]{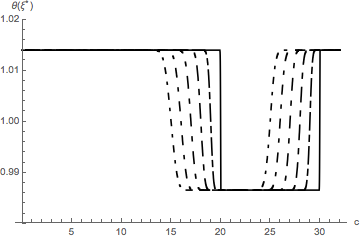}
\includegraphics[height=.25\textwidth]{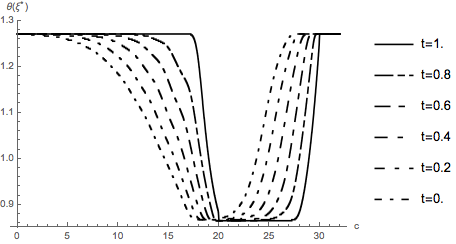}
\caption{Risk loading for the optimal market share $\xi^\ast$  in the SC-model (left) and in the CC-model (right).
(Units on the $c$ axis are $10^6$ units of currency.)}
\label{fig:control}
\end{figure}

We observe the same effects when comparing our CC-model to the SC-model for the example of \cite{LeoNgare}, where a bounded claim size distribution was used.

\paragraph{Concluding remarks}
The introduction of a derivative serves as an effective alternative to classical reinsurance and leads to significantly smaller insurance premia.
The CC-model introduced in this paper has a significant impact on the price of a CAT derivative as well as on the optimal average risk loading. It reflects catastrophes more accurately.

\subsection{Risk management}

In this section we compute the profit-loss distribution and the residual risk of an insurance company holding a CAT derivative.
The former is useful for the derivation of coherent risk measures, the latter quantifies the efficiency of the hedge.

\paragraph{Profit-loss distribution}
The profit-loss distribution is the distribution of the (optimally controlled) wealth $\rho:=X_T^{\xi^\ast}+\psi(C_T)-p$ in case the company holds a CAT derivative.\\

Let $\xi^\ast$ be the optimal control. For $\varsigma\in\mathbb{C}$ define
\begin{align}
\label{eq:laplacevalue}
V_{\varsigma}(t,x,c)=\E\left(\exp\left(-\varsigma\left(X_T^{\xi^\ast}+\psi(C_T)-p\right)\right)\Big|C_t=c,X^{\xi^\ast}_t=x\right)\,.
\end{align}
The right-hand side in \eqref{eq:laplacevalue} can be interpreted as (two-sided) Laplace transform $\hat{L}$ of the profit-loss distribution at $\varsigma$,
$\hat{L}(\varsigma)=\E\left(\exp(-\varsigma \rho)|C_t=c,X^{\xi^\ast}_t=x\right)$.
We can compute $V_{\varsigma}$ by making the ansatz $V_{\varsigma}(t,x,c)=u(x)e^{-\varsigma W_{\varsigma}(t,c)}$, where $W_\varsigma$ solves the backward equation
\begin{equation}\label{eq:bw-w-alpha}
\begin{aligned}
\frac{\partial w_{\varsigma}}{\partial t}(t,c)+q(\xi^\ast)+\tilde \cA^{\xi^\ast}_\varsigma w(t,c)&=0\,,\\
w_{\varsigma}(T,c)&=\psi(c)\,,
\end{aligned}
\end{equation}
with corresponding $\tilde \cA^{\xi^\ast}_\varsigma$.
By solving \eqref{eq:bw-w-alpha} for different values of $\varsigma$, we get the
density $\nu$ of the profit-loss distribution by inverting the Laplace
transform: 
\begin{align*}
\nu(\rho)=\frac{e^{c \rho}}{\pi}\int_{0}^{\infty} \re(\hat L(c+i u))\cos(\rho u)-\im(\hat L(c+i u))\sin(\rho u) du\,.
\end{align*}

\paragraph{Residual risk}
The numerical experiments in Section \ref{subsec:exp} showed that the optimal market share $\xi^\ast$
of an insurance company that holds a CAT derivative is higher than for a company that does not ($\xi^0$).
The change in the strategy $\xi^\ast-\xi^0$ when an insurance company buys a CAT derivative,
is also the strategy used for hedging the derivative itself.

The (buyer's) \textit{risk of derivative} is $\psi(C_T)-p$; the
\textit{residual risk}, i.e.~the remaining risk after hedging is given by 
\begin{align}\label{eq:resrisk}
\psi(C_T)-p+X_T^{\xi^\ast}-X_T^{\xi^0}\,.
\end{align}
The density of \eqref{eq:resrisk} can be computed in the same way as the density of the profit-loss distribution above.


\section*{Acknowledgements}

A. Eichler is supported by the Austrian Science Fund (FWF): Project P21196.

G. Leobacher is supported by the Austrian Science Fund (FWF): Project F5508-N26, which is part of the Special Research Program "Quasi-Monte Carlo Methods: Theory and Applications" and by the Austrian Science Fund (FWF): Project P21196.
This paper was written while G.~Leobacher was member of the Department of Financial Methematics and Applied Number Theory, Johannes Kepler University Linz, Altenbergerstra\ss{}e 69, 4040 Linz, Austria.

M. Sz\"olgyenyi is supported by the Vienna Science and Technology Fund (WWTF): Project MA14-031.



\vspace{2em}
\centerline{\underline{\hspace*{17.5cm}}}

\noindent Andreas Eichler \\
University of Applied Sciences Upper Austria -- Campus Wels, Stelzhamerstra\ss{}e 23, 4600 Wels, Austria\\
andreas.eichler@fh-wels.at\\

\noindent Gunther Leobacher \\
Department of Mathematics and Scientific Computing, University of Graz, Heinrichstra\ss{}e 36, 8010 Graz, Austria\\
gunther.leobacher@jku.at\\

\noindent Michaela Sz\"olgyenyi \Letter \\
Institute for Statistics and Mathematics, WU Vienna University of Economics and Business, Welthandelsplatz 1, 1020 Vienna, Austria\\
michaela.szoelgyenyi@wu.ac.at


\end{document}